\documentclass[letterpaper, 10 pt, conference]{ieeeconf} 

\IEEEoverridecommandlockouts
\makeatletter

\def\ps@IEEEtitlepagestyle{%
    \def\@oddfoot{\mycopyrightnotice}%
    \def\@evenfoot{}%
}
\def\mycopyrightnotice{%
    {\footnotesize  978-1-4799-6773-5/14/\$31.00 \textcopyright2017 Crown\hfill}
    \gdef\mycopyrightnotice{}
}

\usepackage{color}
\usepackage{enumerate,graphicx,epstopdf}
\usepackage{amsmath,amssymb,bm}
\usepackage{amsthm}
\usepackage{cite}
\usepackage{mathtools}
\usepackage{array}
\usepackage{algorithm,algorithmic}
\usepackage{color}
\usepackage{booktabs}
\usepackage{url}

\newtheorem{ass}{Assumption}
\newtheorem{rmk}{Remark}

\newtheorem{dfn}{Definition}
\newtheorem{lmm}{Lemma}

\newtheorem{thm}{Theorem}

\newcommand{\reals}{\mathbb{R}}

\newcommand\eps[0]{\varepsilon}
\newcommand\Kinf[0]{\mathcal{K}_\infty}
\newcommand\KL[0]{\mathcal{KL}}
\newcommand\K[0]{\mathcal{K}}
\newcommand\limsupk[0]{\underset{k\to\infty}{\overline{\lim}}}



\makeatletter
\newcommand*\titleheader[1]{\gdef\@titleheader{#1}}
\AtBeginDocument{%
  \let\st@red@title\@title%
  \def\@title{%
    \bgroup\normalfont\large\centering\@titleheader\par\egroup
    \vskip1.5em\st@red@title}
}
\makeatother

\title{A Semismooth Predictor Corrector Method for \\Suboptimal Model Predictive Control
\thanks{D. Liao-McPherson and I.V. Kolmanovsky are with the University of Michigan, Ann Arbor. Email:\texttt{\{dliaomcp,ilya\}@umich.edu}. M. M. Nicotra is with the University of Colorado Boulder. Email:\texttt{\{marco.nicotra@colorado.edu\}}. This research is supported by the National Science Foundation Award Number CMMI 1562209.}}

\titleheader{This is the accepted version of the paper:\\
D. Liao-McPherson, M. Nicotra, and I. Kolmanovsky, ``A Semismooth Predictor Corrector Method for Suboptimal Model Predictive Control'', European Control Conference, Naples, Italy, June 25-28 2019}
\author{Dominic Liao-McPherson, Marco M. Nicotra, Ilya V. Kolmanovsky}
\begin{document}

\maketitle

\begin{abstract}
Suboptimal model predictive control is a technique that can reduce the computational cost of model predictive control (MPC) by exploiting its robustness to incomplete optimization. Instead of solving the optimal control problem exactly, this method maintains an estimate of the optimal solution and updates it at each sampling instance. The resulting controller can be viewed as a dynamic compensator which runs in parallel with the plant. This paper explores the use of the semismooth predictor-corrector method to implement suboptimal MPC. The dynamic interconnection of the combined plant-optimizer system is studied using the input-to-state stability framework and sufficient conditions for closed-loop asymptotic stability and constraint enforcement are derived using small gain arguments. Numerical simulations demonstrate the efficacy of the scheme.
\end{abstract}


\section{Introduction}

In Model predictive control (MPC) \cite{grune2017nonlinear,rawlings2009model} a control law is defined by the solution of a finite horizon optimal control problem (OCP). Although MPC can systematically handle nonlinearities and constraints, it can be difficult to implement in applications where computing power is insufficient for solving a constrained non-convex OCP at each sampling instance. Developments in numerical solution methods, especially for linear-quadratic MPC, have enabled the application of MPC to a wide variety of systems, see e.g., \cite[Section 2.6]{mayne2014model} and references therein. However, the application of MPC to systems requiring fast sampling rates remains challenging.

Suboptimal MPC (SOMPC) is an approach for reducing the computational cost of MPC controllers. In SOMPC, instead of solving the OCP to a high precision at each sampling instance, we maintain a guess of the optimal solution and improve it each sampling instance, e.g., by shifting the control sequence or performing one of more iterations of an optimization algorithm. The difference between an ideal model predictive controller and an suboptimal model predictive controller is illustrated in Figure~\ref{fig:subopt_mpc_fig}. The ideal MPC law is a static function, while the SOMPC law is a dynamic compensator which maintains an estimate of the optimal solution of the OCP as its internal state.

The paper\cite{scokaert1999suboptimal} established that, in the presence of a suitable terminal set, any feasible solution of the OCP is stabilizing. The robustness properties of SOMPC were studied in \cite{ALLAN201768,pannocchia2011conditions} which established sufficient conditions on the warmstart to ensure stability of the closed-loop system. 

Stability of SOMPC without any terminal conditions or constraints was studied in \cite{grune2010analysis}. Continuous time SOMPC schemes using gradient type optimization methods were proposed in \cite{graichen2010stability,liao2018embedding}, which also derive sufficient conditions for the stability of the combined system based on the convergence rate of the underlying optimization method. Discrete time gradient based schemes are considered in \cite{graichen2012fixed,steinboeck2017design}. The real-time iteration scheme \cite{diehl2005real} is a well known SOMPC strategy for NMPC wherein a single quadratic program is solved per timestep. Sufficient conditions for stability of the combined plant-optimizer system were established in \cite{diehl2005nominal} in the absence of inequality constraints.

In many cases the generation of points which satisfy the optimality conditions of an OCP can be cast as a parameterized rootfinding problem. This is the approach taken in \cite{paternain2018prediction}, which considers the unconstrained case and \cite{zavala2009advanced} which softens constraints with barriers leading to a smooth rootfinding problem. Both paper consider the robustness of the closed-loop system to disturbances caused by suboptimality; however, the treatment of the optimizer itself as a dynamic system in the loop with the plant was not pursued.

In \cite{liaomcpherson2018} we proposed the semismooth predictor-corrector (SSPC) method which generates solutions of the OCP by tracking the roots of a parameterized nonsmooth rootfinding problem. In this work we apply SSPC to SOMPC, and derive sufficient conditions under which the combined system is asymptotically stable using input-to-state stability (ISS) \cite{jiang2001input} and small gain arguments\cite{jiang2004nonlinear}. The performance of the method is illustrated using numerical simulations.

The contributions of this paper are as follows. First, the SSPC method exhibits second order convergence properties at the cost of a single linear system solve per iteration. This compares favorably with gradient methods, which display first order convergence, and SQP type methods, which obtain second order convergence by solving quadratic programs. In addition, compared to existing work for second order SOMPC methods\cite{diehl2005nominal}, we relax the need for a terminal equality constraint, and consider inequality constraints. Third, the proposed stability proof forgoes the rather restrictive requirement of a monotonically decreasing cost function. This is done by using small-gain arguments to determine under what conditions the interconnected plant-optimizer system is contractive. Finally, we establish sufficient conditions for constraint satisfaction.



\begin{figure}
	\centering
	\includegraphics[width=0.95\columnwidth]{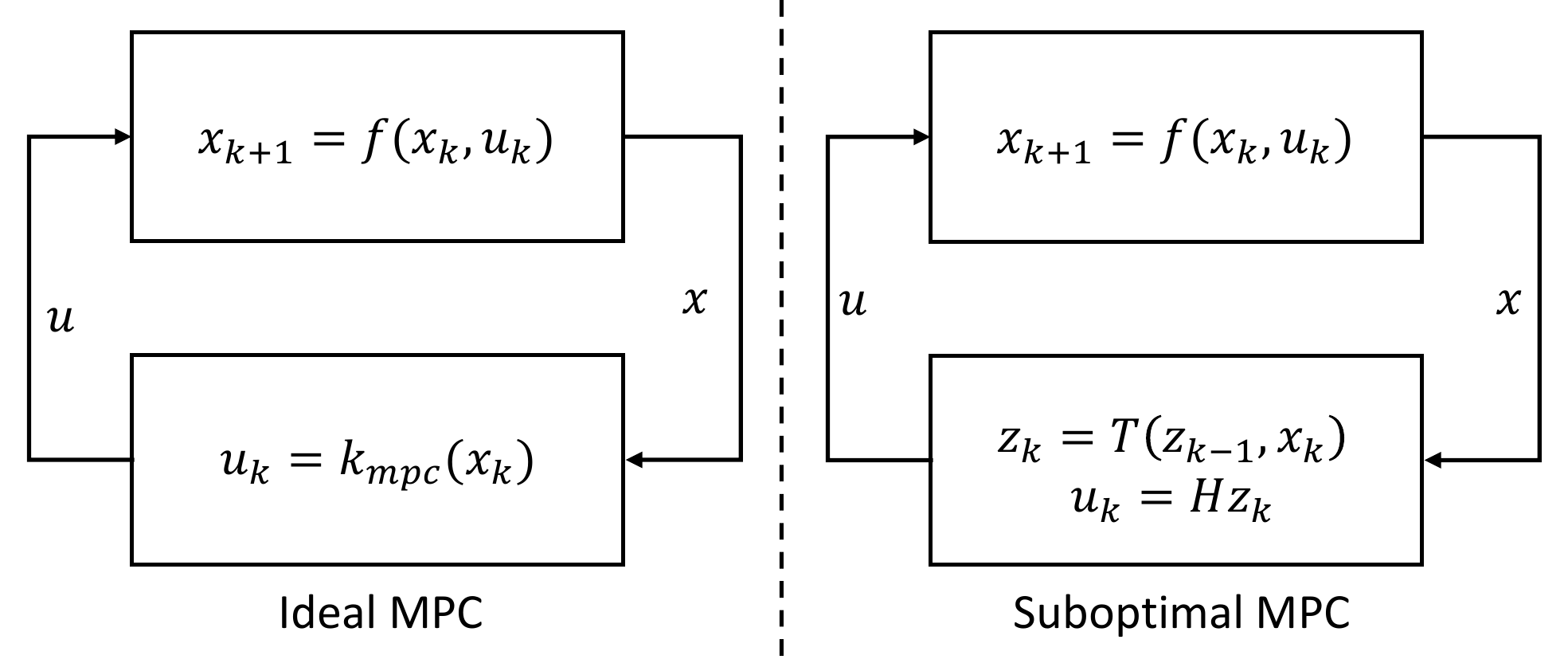}
	\caption{A comparison of suboptimal MPC vs. ideal MPC. Ideal MPC is an implicitly defined static feedback law. Suboptimal MPC can be seen as a dynamic compensator; the current guess of the solution of the OCP is its internal state and an iterative optimization method defines its dynamics.}
	\label{fig:subopt_mpc_fig}
\end{figure}
\section{Problem setting and Control Strategy}
Consider the following discrete-time system,
\begin{equation} \label{eq:dyn_system}
	x_{k+1} = f_d(x_k,u_k),
\end{equation}
where the state and control input are denoted by $x \in X \subset \reals^{n_x}$, and $u\in U \subset \reals^{n_u}$ respectively. 
\begin{ass} \label{ass:ss_properties}
The system satisfies $f_d(0,0) = 0$ and the sets $X$ and $U$ are compact and contain the origin.
\end{ass}

We seek to control \eqref{eq:dyn_system} using MPC, which solves an optimal control problem (OCP) of the form
\begin{subequations} \label{eq:OCP}
\begin{align}
\underset{\xi,u}{\mathrm{min.}}~~&J(\xi,u) = V_f(\xi_N) + \sum_{i = 0}^{N-1} l(\xi_i,u_i),\\
\mathrm{s.t.}~~ &\xi_{i+1} = f_d(\xi_i,u_i), \quad i= 1,\ldots, N-1,\\
& \xi_0 = x, ~~c_N(\xi_N) \leq 0, \label{eq:OCP_param}\\
&c(\xi_i,u_i) \leq 0,  \quad i = 0, \ldots N-1,
\end{align}
\end{subequations}
where $l:\reals^{n_x} \times \reals^{n_u} \to \reals$, $V_f:\reals^{n_x} \to \reals$, $c: \reals^{n_x} \times \reals^{n_u} \to \reals^{n_c}$, $c_N: \reals^{n_x} \to \reals^{n_{cf}}$, at each sampling instance and applies the first element of the solution sequence, $u_0$, to the system. We can then define the ideal MPC feedback law as $k_{mpc}(x) = u_0^*$, where $(\xi^*,u^*)$ is a global minimizer of \eqref{eq:OCP}. Let
\begin{equation}
 	\Gamma = \{x \in X~|~ \text{\eqref{eq:OCP} is feasible}\}
\end{equation} denote the feasible set of \eqref{eq:OCP}. The resulting closed-loop system is
\begin{equation} \label{eq:dyn_cl}
	x_{k+1} = f(x_k,d_k) = f_d(x_k,k_{mpc}(x_k)+d_k),
\end{equation}
where $d_k$ is a disturbance that represents suboptimality. We impose the following conditions on the OCP. The first ensures the existence of second derivatives used by the Newton-type optimization method described in Section~\ref{ss:sspc}, and the second guarantees closed-loop stability of the nominal system (see Theorem~\ref{rmk:MPC_stab}).
\begin{ass} \label{ass:ocp_continuity}
All functions in \eqref{eq:OCP} are $\mathcal{C}^2$ in their arguments.
\end{ass}
\begin{ass} \label{ass:terminal_trio}
The stage cost satisfies $l(0,0) = 0$, and there exists $\alpha_l \in \Kinf$ such that $\alpha_l(||x||) \leq l(x,u)$ for all $(x,u)\in X\times U$. The set $X_f = \{x ~ | ~ c_N(x) \leq 0\}$ is an admissible control invariant set for \eqref{eq:dyn_system} and $V_f$ is a Control Lyapunov Function for \eqref{eq:dyn_system} such that for all $x\in X_f$,
\begin{equation*}
	\underset{u}{\mathrm{min}} \{V_f(x^+) - V_f(x) + l(x,u)~ | ~ (x,u) \in Z, x^+ \in X_f\} \leq 0,
\end{equation*}
where $x^+ = f_d(x,u,0)$ and $Z = \{(x,u) ~ | ~ c(x,u) \leq 0\} \subseteq X\times U$.
\end{ass}
We will also make use of robust positively invariant sets when discussing constraint satisfaction.
\begin{dfn} \cite{limon2009input} The set $\Omega \subseteq \reals^{n_x}$ is a Robust Positively Invariant (RPI) set for system \eqref{eq:dyn_cl} with respect to $D$ if $f(x,d) \in \Omega$ for all $x\in \Omega$, and $d \in D$. If, in addition, $\Omega \subseteq \{x~|~(x,k_{mpc}(x)) \in Z\}$, where $Z$ is defined in assumption~\ref{ass:terminal_trio}, then $\Omega$ is called an admissible RPI set.
\end{dfn} 

Now suppose that not enough computational resources are available to accurately solve \eqref{eq:OCP} at each sampling instance. To reduce computational requirements we can instead approximately track solutions of \eqref{eq:OCP} as the parameter $x$ in \eqref{eq:OCP_param} varies in time by applying a fixed number of iterations of an appropriate numerical iterative method and \textit{warmstarting} each problem with the approximate solution from the previous timestep. This can lead to considerable computational savings. However, it also introduces an error $d_k$ between the optimal MPC control action and the one which is applied. In essence, one has to consider a new dynamical system,
\begin{equation} \label{eq:op_eq}
	z_{k} = T(z_{k-1},x_k),
\end{equation}
where $z$ represents an estimate of the solution of \eqref{eq:OCP} and the function $T$ represents the iterative method which runs in parallel with the plant. This leads to an interconnected plant-optimizer system as shown in Figure~\ref{fig:subopt_mpc_fig}. 

In this paper we suggest a specific Newton-type method and analyze the resulting interconnected system from a systems theoretic point of view. In Section~\ref{ss:background} we provide some background on the concepts used in our analysis. In Section~\ref{ss:sspc} we describe our proposed method in detail. In Section~\ref{ss:ISS_prop} we establish the ISS properties of \eqref{eq:op_eq} and use them to derive sufficient conditions for the stability of the interconnected system in Section~\ref{ss:ISS_MPC}. Finally, numerical examples are reported in Section~\ref{ss:num_ex}.

\section{Background}\label{ss:background}

We will make extensive use of the notion of input-to-state stability (ISS) \cite{jiang2001input}. Since we consider constrained systems it is natural to use a local notion of ISS. We will also make extensive use of the notion of an asymptotic gain\footnote{Recall that a function $\gamma :\reals_+ \to \reals_+$ is said to be of class $\K$ if it is continuous, strictly increasing and $\gamma(0) = 0$. If it is also unbounded then $\gamma \in \Kinf$. A function $\beta:\reals_+ \times \reals_+ \to \reals_+$ is said to be of class $\KL$ if $\beta(\cdot,s) \in \K$ for each fixed $s\geq 0$ and $\beta(r,s) \to 0$ as $s \to \infty$ for fixed $r \geq 0$.}.

\begin{dfn} \label{def:LISS} \cite{jiang2004nonlinear}
Consider a system, 
\begin{equation} \label{eq:ISS_defsys}
	x_{k+1} = f(x_k,u_k),
\end{equation}
and let $\phi(k,x_0,\mathbf{u})$ be its solution at time $k$ with inputs $\mathbf{u} = \{u_0, u_1, ..., u_{k-1}\}$ and initial condition $x_0$. The system is said to be locally input-to-state stable (LISS) if there exists $\eps> 0$, $\beta \in \KL$, and $\gamma \in \K$ such that, $\forall k \in \mathbb{Z}_+$,
\begin{equation}
	||\phi(k,x_0,\mathbf{u})|| \leq \textrm{max}\left[\beta(||x_0||,k),\gamma \left(\underset{0\leq j \leq k}{max}~||u_j||\right)\right],
\end{equation}
provided $||x_0||\leq \eps$ and $\underset{0\leq j \leq k}{max}~||u_j|| \leq \eps$ for all $k \geq 0$.
\end{dfn}
\begin{dfn} \label{def:ISS_gain} \cite{jiang2001input}
Consider system \eqref{eq:ISS_defsys}, we say that it has an asymptotic gain if there exists some $\gamma \in \K$ such that
\begin{equation}
	\limsupk ||\phi(k,x_0,\mathbf{u})|| \leq \gamma \left( \limsupk \left[\underset{0\leq j \leq k}{max}~||u_j||\right]\right),
\end{equation}
for all $x_0 \in \reals^{n_x}$.
\end{dfn}

The following theorem provides tools for establishing conditions under which a system is LISS and for characterizing its asymptotic gain.

\begin{thm} \label{thm:ISS_gain} \cite[Lemma 2.3]{jiang2004nonlinear}
Suppose the system in Definition~\ref{def:LISS} admits a continuous local Lyapunov function $V$ such that for some $\alpha_1,\alpha_2 \in \Kinf$, $c>0$ and $\sigma \in \K$,
\begin{gather*}
\alpha_1(||x||) \leq V(x)\leq \alpha_2(||x||),\\
V(f(x,u)) - V(x) \leq -c V(x) + \sigma(||u||),
\end{gather*}
for all $(x,u)$ in a neighbourhood of the origin. Then the system is LISS and for any $c_0 \in (0,1)$ its asymptotic gain $\gamma \in \K$ can be chosen such that
\begin{equation*}
	\gamma(s) \leq \alpha_1^{-1}\left(\frac{\sigma(s)}{(1-c_0)c}\right),~\forall s > 0.
\end{equation*}
\end{thm}

Next we will impose some conditions on \eqref{eq:OCP} which are needed by the SSPC method. We can compactly write \eqref{eq:OCP} as a parametrized nonlinear program
\begin{subequations}  \label{eq:NLP}
\begin{align}
\underset{w}{\mathrm{min.}} \quad &f(w,p),\\
\mathrm{s.t.} \quad &g(w,p) = 0, \\
&h(w,p) \leq 0,
\end{align}
\end{subequations}
where $w = (\xi,u)$, $x = p$ is the parameter, $f: \reals^n \times \reals^{n_x} \to \reals$, $g:\reals^n \times \reals^{n_x} \to \reals^m$, and $h:\reals^n \times \reals^{n_x} \to \reals^q$. The Lagrangian of \eqref{eq:NLP} is defined as $L(w,\lambda,v,p) = f(w,p) + \lambda^T g(w,p) + v^T h(w,p)$ where $\lambda \in \reals^m$ and $v \in \reals^q$ are dual variables. Let $z = (w,\lambda,v)$, the Karush-Kuhn-Tucker (KKT) conditions for \eqref{eq:NLP} are 
\begin{subequations} \label{eq:KKT}
\begin{gather}
\nabla_w L(w,\lambda,v,p) = 0,\\
g(w,p) = 0,\\
h(w,p) \leq 0,~v \geq 0,~v^T h(w,p) = 0. \label{eq:KKTcomp}
\end{gather}
\end{subequations}
The primal-dual solution mapping of \eqref{eq:KKT}, which may be multivalued since \eqref{eq:NLP} is not assumed convex, will be denoted by 
\begin{equation}
	\bar{S}(p) = \{z = (w,\lambda,v) ~|~ \eqref{eq:KKT} \text{ are satisfied}\}.
\end{equation}
We impose some regularity conditions on the \eqref{eq:NLP} to ensure that the mapping $\bar{S}(p)$ is ``well behaved''. The linear independence constraint qualification (LICQ) is said to hold at a point $(\bar{z},\bar{p})$ if
\begin{equation}
  \text{rank}~\begin{bmatrix}
  \nabla_w g(\bar{w},\bar{p})\\
  [\nabla_w h(\bar{w},\bar{p})]_i
  \end{bmatrix}
   = m + |I_a(\bar{w},\bar{p})|,~ i \in I_a(\bar{w},\bar{p}),
\end{equation}
where $I_a(w,p) = \{i\in 1~...~ q~|~ h_i(w,p) =0\}$ is the index set of active constraints. Further, if a KKT point $(\bar{z},\bar{p})$ satisfying the LICQ also satisfies
\begin{equation}
  u^T \nabla_w^2 L(\bar{z},\bar{p}) u> 0,~\forall u \in \mathcal{K}_+(\bar{w},\bar{v},\bar{p}) \setminus \{0\},
\end{equation}
where $\mathcal{K}_+(w,v,p) = \{u\in \reals^n~|~ \nabla_w g(\bar{w},\bar{p})u = 0,~ \nabla_w h_i(\bar{w},\bar{p}) u \leq 0, i \in I^+_a(\bar{w},\bar{v},\bar{p}), \nabla_w f(\bar{w},\bar{p})^T u \leq 0\}$, and $I_a^+(w,v,p) = I_a(w,p) \cap \{i~|~ v_i > 0\}$ then it is said to satisfy the strong second order sufficient conditions (SSOSC). Any KKT point which satisfies the SSOSC and the LICQ is a strict local minimizer of \eqref{eq:NLP}. Our main regularity assumption follows.
\begin{ass} \label{ass:pointwise_strong_reg} (Pointwise strong regularity)
The LICQ and SSOSC hold at all KKT points in $\Gamma$.
\end{ass}
Theorem~\ref{thm:lipschitz_param} establishes Lipschitz continuity of primal-dual solutions of \eqref{eq:KKT} and of the optimal value function. Theorem~\ref{thm:sol_traj} shows that the solution trajectories of the OCP are isolated and can be tracked. Finally, Theorem~\ref{rmk:MPC_stab} establishes the ISS properties of \eqref{eq:OCP}.
\begin{thm} \label{thm:lipschitz_param}
At each $(\bar{z},\bar{p})$ satisfying \eqref{eq:KKT} there exists a neighbourhood $P$ of $\bar{p}$ and a constant $L_p(\bar{p},\bar{z})$ such that $\bar{S}(p)$ is a single valued function satisfying $||\bar{S}(p) - \bar{z}|| \leq L_p ||p - \bar{p}||,~ \forall p \in P$.
\end{thm}
\begin{proof}
The LICQ and SSOSC are necessary and sufficient for strong regularity of the KKT system, see e.g., \cite[Theorem 2G.8]{dontchev2009implicit}, and strong regularity implies that $S$ is locally a Lipschitz continuous function \cite[Theorem 2B.1]{dontchev2009implicit}.
\end{proof}

\begin{thm} \label{thm:sol_traj}\cite[Theorem 6G.1]{dontchev2009implicit} Suppose $p$ is prescribed as a Lipschitz continuous function of a scalar $t \geq 0$. Then the solution trajectory mapping $S(p(t))$ is comprised of isolated Lipschitz continuous trajectories.
\end{thm}
In this paper we will concern ourselves with 
\begin{equation}
	S(p) \in \bar{S}(p)
\end{equation}
which denotes the solution mapping corresponding to the global optimum \footnote{We assume that $S$ is a function, if it is not then one could consider a restriction of $S$ thanks to Theorem~\ref{thm:sol_traj}} of \eqref{eq:OCP}. Local minima may cause the closed-loop system to converge to non-zero equilibrium points.

\begin{thm} \label{rmk:MPC_stab} \cite[Theorem 4]{limon2009input} 
Let Assumptions~\ref{ass:ss_properties} - \ref{ass:pointwise_strong_reg} hold, then the closed-loop system~\eqref{eq:dyn_cl} is LISS with respect to $d$ on a robust positively invariant set $\Omega \subseteq \Gamma$ .
\end{thm}

\begin{rmk}
Pointwise strong regularity is a common assumption in time varying optimization, e.g., \cite{zavala2010real,hours2016parametric,dinh2012adjoint,dontchev2013euler}. The SSOSC is generally easy to enforce through appropriate regularization \cite{bieglerIFAC}. If the only constraints are upper and lower control input bounds then the LICQ can be proven to hold a-priori, otherwise the problem can be reformulated, e.g., as described in \cite{yang2015nonlinear}.
\end{rmk}

\section{The semismooth predictor corrector method} \label{ss:sspc}
This section describes the SSPC method introduced in \cite{liaomcpherson2018} which is based on mapping the KKT conditions to a parameterized rootfinding problem by replacing the complementarity conditions \eqref{eq:KKTcomp} with nonsmooth equations. This is done using an nonlinear complementarity (NCP) function \cite{sun1999ncp} 
$\psi:\reals^2 \mapsto \reals$ which has the property that
\begin{equation}
	\psi(a,b) = 0 \Leftrightarrow a \geq 0,~b\geq 0, ~ab = 0.
\end{equation}
We use the Fischer-Burmeister (FB) NCP function \cite{fischer1992special}
\begin{equation}
	\psi(a,b) = a+b - \sqrt{a^2 + b^2}.
\end{equation}
The FB function is semismooth \cite{qi1993nonsmooth}, which allows us to use a nonsmooth Newton-type method based on generalized derivatives\footnote{For a function $G:\reals^N \mapsto \reals^M$, $\partial G(x)$ denotes Clarke's Generalized Jacobian \cite{clarke1990optimization}.}. Following \cite{qi1997semismooth,fischer1992special} the SSPC method uses the FB function to map points satisfying \eqref{eq:KKT} to roots of the following parameterized semismooth rootfinding problem,
\begin{equation} \label{eq:Fmapping}
	F(z,p) = \begin{bmatrix}
		\nabla_z L(w,\lambda,v,p)\\
	g(w,p)\\
	\phi(-h(w,p),v)
	\end{bmatrix},
\end{equation}
where $z = (w,\lambda,v)$ is the primal-dual variable and $\phi$ is the concatenation of $\psi(-h_i(w,p),v_i)$ for $i = 1,\ldots,q$. For any fixed $p$ the roots of $F(z,p)$ coincide with $\bar{S}(p)$ so we can obtain solution trajectories by tracking solutions of $F(x,p) = 0$ as $p$ varies in time. The predictor and corrector steps are,
\begin{subequations} \label{eq:pc}
\begin{gather}
V_{k-1} (p_k - p_{k-1}) + B_{k-1} (\bar{z}_k - z_{k-1}) = 0, \label{eq:predictor}\\
F_{k}(\bar{z}_k,p_k) + \bar{B}_{k} (z_k - \bar{z}_k) = 0, \label{eq:corrector}
\end{gather}
\end{subequations}
where $V_{k-1} \in \partial_p F(z_{k-1},p_{k-1})$, $B_{k-1} \in \partial_zF(z_{k-1},p_{k-1})$, and $\bar{B}_k = \partial_z F(\bar{z}_k,p_k)$. The predictor solves \eqref{eq:predictor} for $\bar{z}_k$. The product $B_{k-1}^{-1}V_{k-1}\Delta p$ is the directional derivative of the solution mapping in the direction $\Delta p = p_k-p_{k-1}$. The predictor can thus be interpreted as an Euler integration step. The corrector solves \eqref{eq:corrector} for $z_k$ and is a single iteration of the semismooth Newton's method. The generalized Jacobians used in \eqref{eq:pc} are 
\begin{equation} \label{eq:dF_z}
	\partial_z F = \begin{bmatrix}
		\nabla_w^2 L(z,p) & \nabla_w g(w,p)^T & \nabla_w h(w,p)^T\\
		\nabla_w g(w,p) & 0 & 0\\
		-C \nabla_w h(w,p) & 0 & D
	\end{bmatrix},
\end{equation}
where $C = \text{diag}(\nu)$ and $D = \text{diag}(\mu)$ are
\begin{equation}
	(\nu_i,\mu_i) \in \begin{cases}
	(1  + \frac{h_i}{r_i},1  - \frac{v_i}{r_i}), & \text{if } (h_i,v_i) \neq 0,\\
	(1-a,1-b),  & \text{if } (h_i,v_i) = 0,
	\end{cases},
\end{equation}
$h_i = h_i(w,p)$, $r_i = ||(h_i(w,p),v_i)||_2$, and $(a,b)$ are arbitrary scalars satisfying $||(a,b)||_2 = 1$. Any value of $(a,b)$ works, we use $a = b = 2^{-1/2}$; in our experience adjusting this value does not result in any performance changes. The derivative $\partial_p F(z,p)$ consists of all matrices of the form
\begin{equation} \label{eq:dF_p}
	\partial_p F= \begin{bmatrix}
		\nabla_{pz}L(z,p)\\
		\nabla_{p}g(w,p)\\
		-C \nabla_{p} h(w,p)
	\end{bmatrix},
\end{equation}
where $C$ is the same matrix as in \eqref{eq:dF_z}.  All elements of $\partial_z F(z,\bar{p})$ are guaranteed to be invertible in a neighbourhood of any $\bar{z} \in S(\bar{p})$ \cite[Proposition 1]{liaomcpherson2018}. As detailed in \cite[Theorem 2]{liaomcpherson2018} one can establish error bounds for \eqref{eq:pc},
\begin{subequations} \label{eq:sspc_bounds}
\begin{gather} 
	||\bar{e}_k|| \leq ||e_{k-1}|| + c ||p_k - p_{k-1}||^2,\\
	||e_k|| \leq \eta ||\bar{e}_k||^2,
\end{gather}
\end{subequations}
where $e_k = z_k - S(p_k)$, and $c, \eta >0$ are positive constants that depend on the properties of $F$. The error bound \eqref{eq:sspc_bounds} is looser than the one given in \cite[Theorem 2]{liaomcpherson2018} but is algebraically cleaner. The proof is analogous to that of \cite[Theorem 2]{liaomcpherson2018}.

\begin{rmk} \label{rmk:err_bnd}
The error bound \eqref{eq:sspc_bounds} holds provided $\Delta p_{k-1} = p_k - p_{k-1} \in P(p_{k-1})$ and $e_k \in E(p_k)$ where $P(p), E(p)$ are set valued mappings to neighbourhoods of the origin. Since the parameter set $X$ is compact there exists sets $\mathcal{E}$ and $\mathcal{P}$ satisfying $\mathcal{E} \subseteq E(p),~ \mathcal{P} \subseteq P(p),~\forall p \in \Gamma \subseteq X$.
\end{rmk}

\section{ISS properties of the SSPC method} \label{ss:ISS_prop}
In this paper we will consider a variant of the SSPC method where $\ell \in \mathbb{Z}_{>0}$ corrector steps are taken. We can view this process as the following dynamic system,
\begin{subequations}
\begin{gather}
	z_k = T_\ell(z_{k-1}, x_k), \\ 
	u_k = H z_k
\end{gather}
\end{subequations}
where  $H$ is the matrix which selects the control input from the primal-dual solution, i.e., $u_k = HS(x_k) = k_{mpc}(x_k),$
so that the output $u_k$ approximates $k_{mpc}(x_k)$. For our ISS analysis we will work with the associated error system\footnote{The dynamic equation of the error system has been shifted forward by one time instance to bring it into the form typically used in ISS analyses.},
\begin{gather} \label{eq:error_system}
	e_{k+1} = G_\ell(e_k,\Delta x_k),\\
	\Delta u_k = H e_k
\end{gather}
where $\Delta x_k = x_{k+1} - x_k$. The error system obeys the following property
\begin{equation} \label{eq:err_bound_iss_start}
	||e_{k+1}|| \leq \eta^{2^\ell -1}(||e_k|| + c||\Delta p_k||^2)^{2^\ell},
\end{equation}
which is obtained from \eqref{eq:sspc_bounds} with the corrector applied $\ell$ times. The main result of this section is that \eqref{eq:error_system} is LISS with $\Delta x$ as an input and that, under certain conditions, its ISS gain approaches 0 as $\ell \to \infty$. We begin with a technical lemma which is proven in the appendix. 

\begin{lmm} \label{lmm:pow2}
For any scalars $a,b \geq 0, k \in \mathbb{N}$ the following holds: $(a + b)^{2^k} \leq 2^{2^{k-1}} (a^{2^k} + b^{2^k})$.
\end{lmm}

\begin{thm} \label{thm:ISSproof}
Consider the SSPC error system \eqref{eq:error_system} and let Assumptions~\ref{ass:ocp_continuity} and \ref{ass:pointwise_strong_reg} hold. Then there exists a monotonically increasing function $\eps_1(\ell) >0$, $\eps_2 > 0$, and $\rho \in \Kinf$ such that \eqref{eq:error_system} is LISS if $||e_0|| < \text{min}(\eps_1(\ell),\eps_2)$. $2\eta\cdot\text{min}(\eps_1(\ell),\eps_2) < 1$, and $||\Delta \mathbf{x}_k|| = \underset{0\leq j \leq k}{max}~||\Delta x_j|| \leq \rho^{-1}(\text{min}(\eps_1(\ell),\eps_2))$ $\forall k \in \mathbb{Z}_+$. In addition, there exists $\eps_3$ such that if $||\Delta \mathbf{x}_k|| < \eps_3,~\forall k \in \mathbb{Z}_+$ then there exists $\gamma_2 \in \KL$ such that the ISS gain of \eqref{eq:error_system} satisfies $\gamma(s) \leq \gamma_2(s,\ell)$.  
\end{thm}
\begin{proof}
We begin with \eqref{eq:err_bound_iss_start}; using Lemma~\ref{lmm:pow2} and performing some algebraic manipulations we obtain,
\begin{align}
||e^+|| &\leq \eta^{2^\ell -1}(||e|| + c||\Delta x||^2)^{2^\ell},\\
 & \leq \frac{1}{2\eta} (2\eta||e||)^{2^\ell} + \frac{1}{2\eta} (2c\eta)^{2^\ell}||\Delta x||^{2^\ell +1},
\end{align}
where $e^+,e$, and $\Delta x$ are shorthand for $e_{k+1},e_k$, and $\Delta x_k$. Consider the candidate ISS Lyapunov function $||\cdot|| \in \Kinf$. After further algebraic manipulations, we have that
\begin{multline}
||e^+||- ||e|| \leq -||e|| \left[1 - \frac{1}{2\eta}(2\eta||e||)^{2^\ell}\right] +\\
\frac{1}{2\eta} (2c\eta||\Delta x||)^{2^\ell}||\Delta x||^2,
\end{multline}
\begin{equation}
 = -\alpha(||e||,\ell) + \sigma(||\Delta x||,\ell).
\end{equation}
Consider the function $\alpha(r,\ell)$, by analyzing the equation
\begin{equation}
	\frac{\partial\alpha}{\partial r}(r,\ell) = \alpha'(r,\ell) = 1- \frac{1}{2\eta} (1+2^\ell)(2\eta r)^{2^\ell} = 0,
\end{equation}
we see that, for a fixed $\ell$, $\alpha$ is increasing and thus of class $\K$ on the domain $\mathcal{D} = [0,\bar{\eps}(\ell))$ where
\begin{equation}
	\bar{\eps}(\ell) = (2\eta)^{\frac{1}{2^\ell}-1} \left(\frac{1}{1+2^\ell}\right)^{\frac{1}{2^\ell}}.
\end{equation}
For any fixed $\ell > 0$ the term
\begin{equation}
	a(||e||,\ell) = 1 - \frac{1}{2\eta}(2\eta||e||)^{2^\ell},
\end{equation}
is positive on the interval $[0,(2\eta)^{\frac{1}{2^\ell} -1}) \subseteq \mathcal{D}$ and thus satisfies the inequality
\begin{equation}
	0 \leq a(\bar{\eps},\ell) \leq a(r,\ell) \leq a(0,\ell),~~ \forall r \in \mathcal{D}.
\end{equation}
Pick an arbitrary scalar $\tau \in (0, \bar{\eps})$, then
\begin{align}
	||e^+||- ||e|| &\leq -\alpha(||e||,\ell) + \sigma(||\Delta x||,\ell),\\
	& \leq -\bar{a}||e|| + \sigma(||\Delta x||,\ell),\\
	& = -\bar{a}||e|| + \sigma(||\Delta x||,\ell),
\end{align}
where $\bar{a} = a(\bar{\eps} - \tau,\ell)$. Its clear that $\sigma \in \mathcal{K}_\infty$ for any fixed $\ell$ thus the dissipation inequlity in Theorem~\ref{thm:ISS_gain} holds provided $||e_k|| < \bar{\eps}(\ell)$. Applying Theorem~\ref{thm:ISS_gain} we obtain that, for any $c_0 \in (0,1)$, the ISS gain can be chosen to satisfy $\gamma(s,\ell) \leq  c_1 \sigma(s,\ell)$, where $c_1^{-1} = (1-c_0)\bar{a}$. Since \eqref{eq:error_system} is LISS for $||e_k||$ sufficiently small, we can recursively enforce that $||e_k|| < \bar{\eps}(\ell)~\forall k\in \mathbb{Z}_+$ by restricting $||e_0||\leq \bar{\eps}(\ell)$ and $\bar{a}^{-1}\sigma(||\Delta \mathbf{x}_k||,\ell) \leq \bar{\eps}(\ell)$ \cite[Remark 3.7]{jiang2001input}. Similarly, to enforce that $(e_k,\Delta x_k) \in \mathcal{E}\times \mathcal{P},~\forall k\in \mathbb{Z}_+$, (see Remark~\ref{rmk:err_bnd}) we also restrict the initial condition and input to satisfy $||e_0||\leq \eps_2$ and $\bar{a}^{-1}\sigma(||\Delta \mathbf{x}_k||,\ell) \leq \eps_2$ where $\eps_2$ is chosen small enough so that a ball of radius $\eps_2$ is contained within $\mathcal{E}$ and $\mathcal{P}$. This is possible since they are neighbourhoods of the origin. In addition, if $2c\eta r < 1$ then $\sigma(r,\ell) \in \KL$. Letting $\eps_1(\ell)= \bar{\eps}(\ell)$, $\eps_3 = \frac{1}{2c\eta}$, and $\rho = \bar{a}^{-1}\sigma$  completes the proof.
\end{proof}

\section{ISS properties of suboptimal MPC} \label{ss:ISS_MPC}
Theorem~\ref{thm:ISSproof} illustrates that, under some conditions, the SSPC method, viewed as a dynamic system driven by parameter changes, is LISS and has an asymptotic gain which can be made arbitrarily small by performing more iterations. Since the ideal closed loop system \eqref{eq:dyn_cl} is LISS, we can treat the sub-optimality error as a disturbance and derive sufficient conditions for the stability of the interconnection between the SSPC and the plant using small gain arguments.
\begin{thm} \label{prp:small_gain}
Consider the interconnected dynamic systems
\begin{subequations}  \label{eq:iss_sys}
\begin{equation} \label{eq:sys1}
	\Sigma_1:\begin{cases}
	~~ x_{k+1} = f(x_k,d_k),\\
	~~ \Delta x_k = h(x_k,d_k)
	\end{cases}
\end{equation}
\begin{equation} \label{eq:sys2}
	\Sigma_2:\begin{cases}
	~~e_{k+1} = G_\ell(e_k,\Delta x_k),\\
	~~ d_k = He_k
	\end{cases}
\end{equation}
\end{subequations}
where $f(x,d) = f_d(x,k_{mpc}(x)+d)$, $f_d$ is defined in \eqref{eq:dyn_system},  $h(x,d) = f(x,d) - x$, and $G_\ell$ is defined in \eqref{eq:error_system}. Let Assumptions \ref{ass:ss_properties} - \ref{ass:pointwise_strong_reg} and the assumptions of Theorem~\ref{thm:ISSproof} hold. Then there exists $\ell^* > 0$ such that, if $\ell\geq\ell^*$, the origin is an asymptotically stable equilibrium point for \eqref{eq:iss_sys} whose region of attraction satisfies $\mathcal{R} \subset \Gamma \times \mathcal{E}$.
\end{thm}
\begin{proof}
Under Assumptions \ref{ass:ss_properties}-\ref{ass:pointwise_strong_reg} $\Sigma_1$ is LISS by Theorem~\ref{rmk:MPC_stab} and thus admits an asymptotic gain from $d$ to $x$ \cite[Lemma 3.8]{jiang2001input}. Since the output equation is Lipschitz continuous this implies that there exists $\gamma_1 \in \K$ such that 
\begin{equation}
	\limsupk ||\Delta x_k|| \leq \gamma_1\left(\limsupk ||d_k||\right).
\end{equation}
Similarly by Theorem~\ref{thm:ISSproof} we have that $\Sigma_2$ is LISS and that there exists $\gamma_2\in \KL$ such that 
\begin{equation}
 	\limsupk ||d_k|| \leq \gamma_2\left(\limsupk||\Delta x_k||,\ell\right).
\end{equation} 
Thus using small gain arguments, see e.g., \cite{jiang2004nonlinear}, we have that \eqref{eq:iss_sys} is contractive in a neighbourhood of the origin provided $\gamma_1 \circ \gamma_2(s,\ell) < s, ~\forall s \geq 0$. Since $\gamma_2 \in \KL(s,\ell)$ it can be made arbitrarily small by letting $\ell \to \infty$. It follows from the finiteness of $\gamma_1$ that there exists $\ell^*$ such that the small gain condition is satisfied. To conclude the proof we define $\mathcal{R}$ as the set of initial conditions under which $(x_k,e_k) \in \Gamma \times \mathcal{E}$ for all $k\in \mathbb{Z}_+$.
\end{proof}
Theorem~\ref{prp:small_gain} establishes asymptotic stability of the interconnected plant-optimizer system but doesn't consider constraint satisfaction. Since the closed-loop system under ideal feedback law is LISS on an admissible RPI set we can derive sufficient conditions for constraint satisfaction that are summarized in the following theorem.

\begin{thm} \label{thm:constraint_satisfaction}
Suppose that the assumptions of Theorem~\ref{prp:small_gain} hold so the interconnected system \eqref{eq:iss_sys} is LISS. Let $\Omega$ denote the admissible RPI set in Theorem~\ref{rmk:MPC_stab}, let $\gamma_2(s,\ell) \in \KL$ upper bound $\gamma(s)$, the asymptotic gain of \eqref{eq:sys2}, and let $(x_k,e_k)$ denote the closed-loop trajectory of \eqref{eq:iss_sys} for some initial condition $(x_0,e_0)$. Then there exists $\bar{\ell} \geq \ell^*$ and $\delta > 0$ such that if $||e_0|| \leq \delta$, and $x_0\in \Omega$ then $x_k \in \Omega$ for all $k\geq 0$.
\end{thm}
\begin{proof}
By Theorem~\ref{rmk:MPC_stab}, there exists a neighbourhood $D$ of the origin such that, if $d_k \in D, ~\forall k \geq 0,$ and $x_0 \in \Omega$, then $x_k \in \Omega,~\forall k\geq 0$. Since $d = H e$ for a fixed matrix $H$, this implies the existence of $\rho> 0$ such that if $||e_k|| \leq \rho$ then $d_k \in D$. Further, \eqref{eq:iss_sys} is LISS, so we can enforce $||e_k|| \leq \rho$ by noting that,
\begin{equation}
	||e_k|| \leq \max\left\{\beta(||e_0||,k), \gamma_2\left(\underset{0\leq j \leq k}{max}~||\Delta x_j||,\ell\right)\right\},
\end{equation}
by the definition of LISS. Thus we must impose that 
\begin{equation}
	||e_0|| \leq \beta_0^{-1}(\rho),
\end{equation}
where $\beta_0(\cdot) = \beta(\cdot,0)\in \K$, and
\begin{equation}
	\gamma_2\left(\underset{0\leq j \leq k}{max}~||\Delta x_j||,\ell\right) \leq \rho
\end{equation}
for all $\Delta x \in \Delta \Omega = \Omega - \Omega$. As proven in \cite[Theorem 4]{limon2009input}, the set $\Omega$ is bounded, implying that there exists
\begin{equation}
	\bar{s} = \underset{w\in \Delta \Omega}{\sup}~||w|| < \infty.
\end{equation}
Since $\bar{s}$ is finite and $\gamma_2 \in \KL$ there must exist some $\ell_1$ such that $\gamma_2(\bar{s},\ell_1) \leq \rho$. Letting $\delta = \beta_0^{-1}(\rho)$ and $\bar{\ell} = \max(\ell^*,\ell_1)$ completes the proof.
\end{proof}

Theorem~\ref{thm:constraint_satisfaction} establishes that, if enough computational resources are available and the initial solution guess is sufficiently accurate, then constraints are guaranteed to be satisfied.

\section{Numerical examples} \label{ss:num_ex}
In this section we illustrate our theoretical results with a numerical example. The attitude dynamics of a rigid spacecraft are given by the Euler equations,
\begin{equation}
	\dot{x} = f_c(x,u) = \begin{bmatrix}
		J^{-1}(-\omega^\times J \omega + u)\\
		S(\theta) \omega
	\end{bmatrix},
\end{equation}
where $\omega\in \reals^3$ is the vector of angular velocities expressed in a body fixed frame, $\theta$ is the vector or 3-2-1 Euler angles, $x = [\omega^T~\theta^T]^T$ is the state vector, $J = diag(918,920,1365)$, is the inertia matrix, $u \in \reals^3$ are external control moments and
\begin{equation}
	S(\theta) = \begin{bmatrix}
		1 & \sin(\theta_1) \tan(\theta_2) & \cos(\theta_1)\tan(\theta_2) \\
		0 & \cos(\theta_1) & -\sin(\theta_1)\\
		0 & \sin(\theta_1) \sec(\theta_2) & \cos(\theta_1) \sec(\theta_2)
	\end{bmatrix}.
\end{equation} The control objective is to drive the system from $x(0) = [0~~0~~0~~15^\circ~~30^\circ~~-20^\circ]$ to the origin. We discretize the dynamics using explicit Euler integration with a timestep of $\tau = 3~s$. The system is placed in closed-loop with a suboptimal MPC controller implemented using SSPC. The optimal control problem \eqref{eq:OCP} to be solved at time $t_k$ is
\begin{subequations}
\begin{align}
\underset{\xi,u}{min.}~~&||\xi_N||_P^2  + \sum_{i=0}^{N-1} ||\xi_i||_Q^2 + ||u||_R^2\\
s.t. ~~& \xi_{i+1} = f_p(\xi_i,u_i), ~ i = 0, \ldots ,N-1,\\
&|\omega_i| \leq 0.02, ~ i = 1,\ldots,N,\\
&|u_i| \leq 2, ~ i = 0,\ldots, N-1,\\
&A_f \xi_N \leq b_f, ~~ \xi_0 = x(t_k)
\end{align}
\end{subequations}
where $f_p(x,u) = x + \tau f_c(x,u)$, and $N= 30$ is the prediction horizon. We obtain $P$ as the solution of the discrete time algebraic Riccati equation with the dynamics linearized about the origin\footnote{$Q = 50~diag([10~10~10~1~1~1])$, $R = 0.1 I_{3\times 3}$}. The terminal control invariant set $X_f = \{x~|~ A_f x \leq b_f\}$ is computed using the MPT3 toolbox \cite{MPT3}, derivatives are computed using CASADI\cite{Andersson2018} and the solution estimate is initialized at the origin \footnote{The initial guess has been purposefully chosen to be poor, the initial residual is on the order of 1000.}.

Closed-loop simulation results are shown in Figure~\ref{fig:sim_traces}. Only one corrector iteration per timestep is needed for stability, however constraints are not satisfied. When $\ell$ is increased to 2 both the state and control constraints are satisfied as predicted by Theorem~\ref{thm:constraint_satisfaction}. Figure~\ref{fig:residuals} shows the KKT residual $||F(z,x)||$, which upper and lower bounds the error \cite{qi1997semismooth}, and the closed-loop cost function. The maximum wall clock execution time\footnote{Performed using default settings on a 2015 Macbook pro with a 2.7 GHz i7 CPU and 16 GB of RAM running MATLAB 2017b.} for SSPC, implemented in native MATLAB code, for the $\ell = 2$ case was $0.0371 s$, compared to $1.41 s$ for \texttt{fmincon} and $0.9701 s$ for \texttt{ipopt}.

\begin{figure}
	\centering
	\includegraphics[width=0.95\columnwidth]{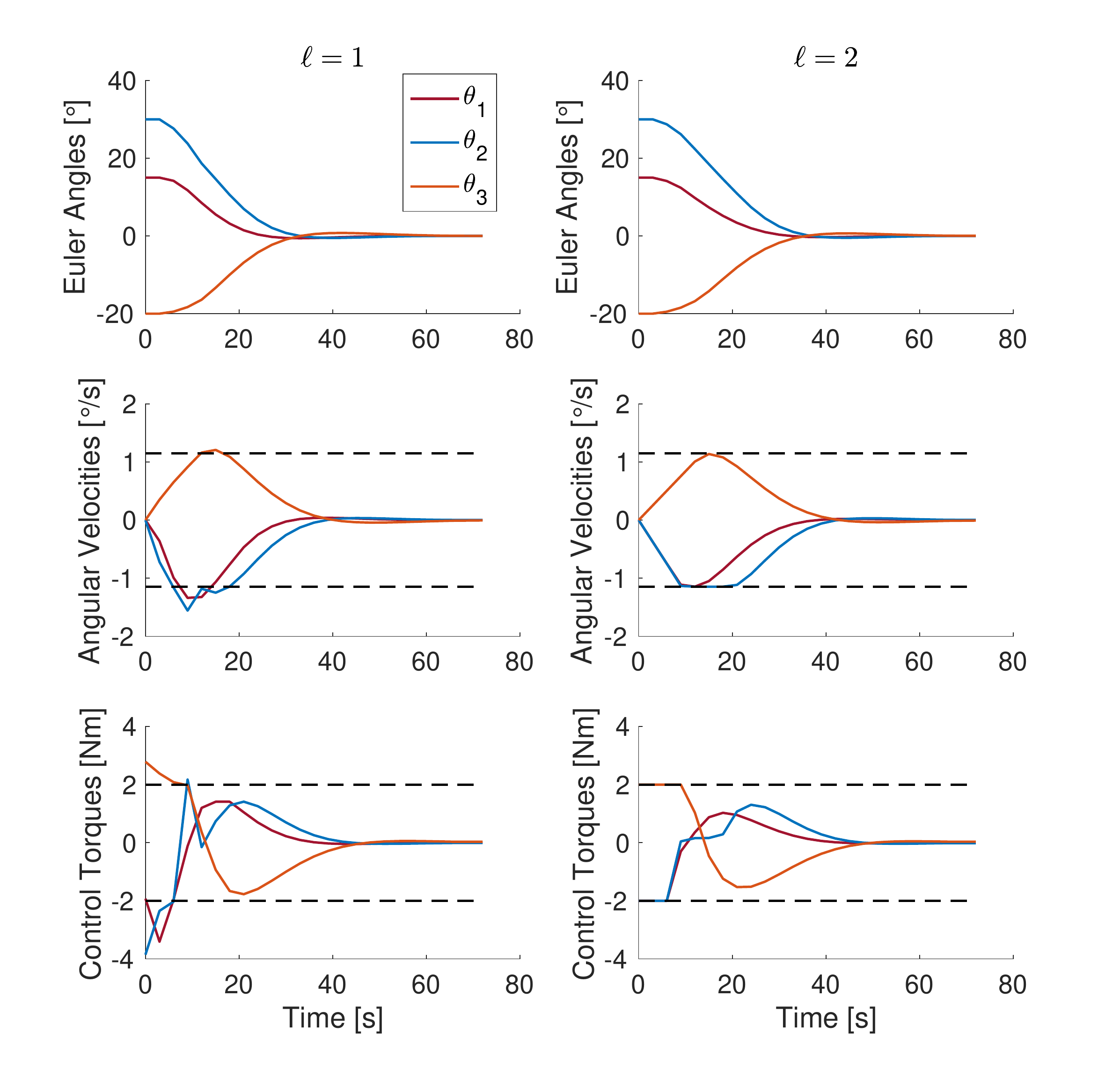}
	\caption{Simulated closed-loop trajectories of the spacecraft in loop with suboptimal MPC with one iteration (left) and two iterations (right).}
	\label{fig:sim_traces}
\end{figure}

\begin{figure}
	\centering
	\includegraphics[width=0.95\columnwidth]{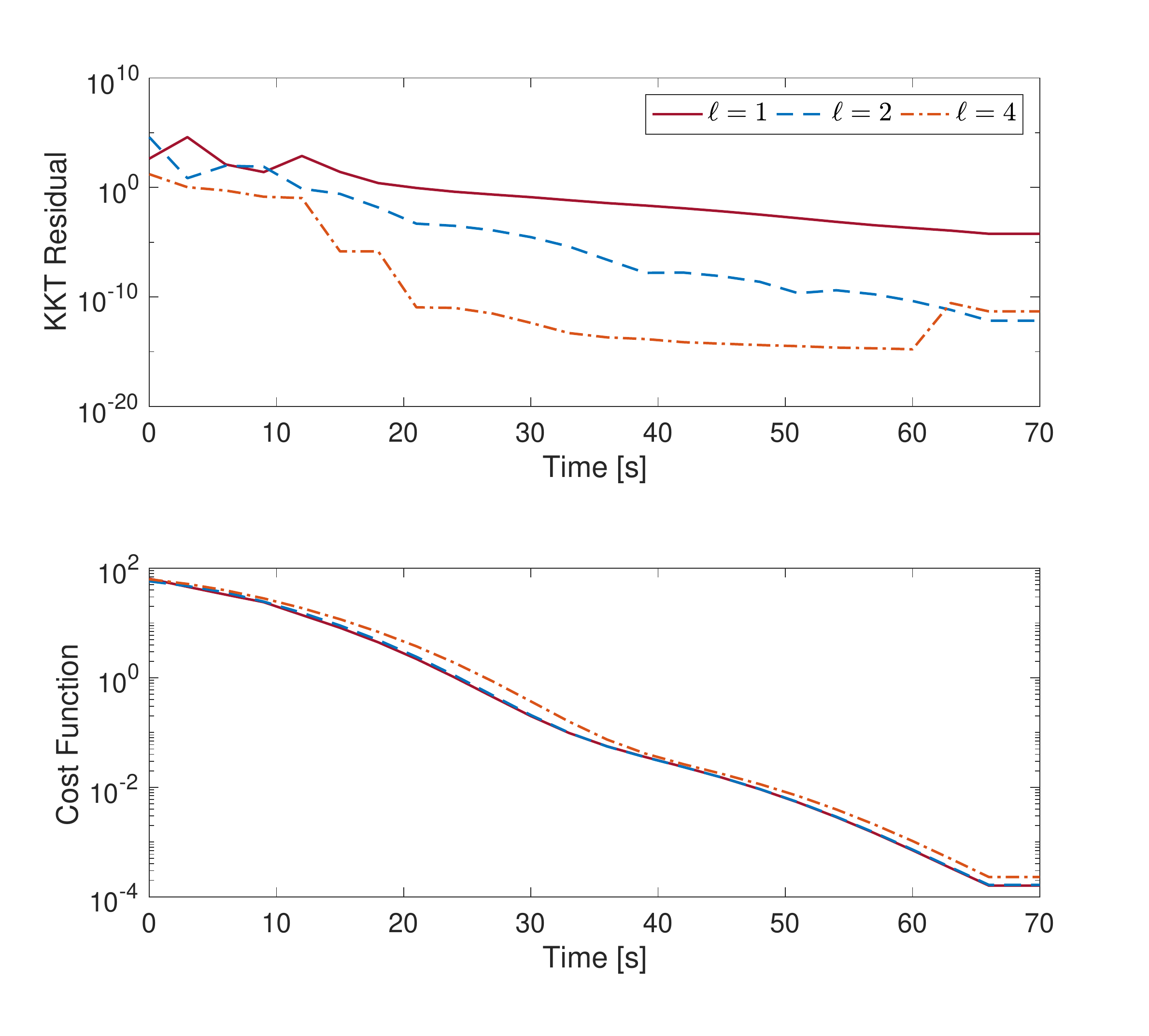}
	\caption{Closed-loop residual and MPC cost function for 1,2, and 4 Newton iterations per timestep. The residual in the $\ell = 1$ case decays to machine precision after about 250 s. }
	\label{fig:residuals}
\end{figure}
\section{Conclusions}
In this paper we introduced a suboptimal MPC method based on the SSPC algorithm. We established conditions under which the SSPC algorithm, viewed as a dynamic system, is LISS. We also establish sufficient conditions for stability of the combined SSPC-plant system using small gain theorem based arguments and sufficient conditions for constraint satisfaction. Numerical simulations show that stability can be achieved with only one corrector iteration per time step even in the presence of a large initial estimate error. Constraint enforcement requires either more iterations or a better initial estimate. Future work includes an investigation of the effect of the sampling period, and the design of robust MPC controllers tailored for implementation using SSPC.
\bibliography{iss_mpc}

\appendix
\textbf{Proof of Lemma~\ref{lmm:pow2}:} The proof is by induction. The base case is $(a+b)^2 \leq 2 (a^2 + b^2)$ which holds by the Cauchy-Shwartz inequality. Let $(a + b)^{2^k} \leq 2^{2^{k-1}} (a^{2^k} + b^{2^k})$ hold for all $k \in \mathbb{N}$. Then for $k+1$ we have that
\begin{align*}
	(a+b)^{2^{k+1}}&= (a+b)^{2^{2^k}} \leq (2(a^2+b^2 ))^{2^k},\\
	&=2^{2^k} (a^2+b^2)^{2^k} \leq 2^{2^k} 2^{2^{k-1}} ((a^2 )^{2^k}+(b^2 )^{2^k})\\
	&\leq 2^{2^k}(a^{2^{k+1}} + b^{2^{k+1}}).
\end{align*}

\end{document}